\documentclass[12pt]{article}

\usepackage{geometry}
\usepackage{alg}
\usepackage{cite}
\usepackage{booktabs}
\usepackage{paralist}
\usepackage{graphicx}
\usepackage{subfigure}
\usepackage{amssymb}
\usepackage{amsmath}
\allowdisplaybreaks[1]

\usepackage{amsthm}

\newtheorem{theorem}{Theorem}
\newtheorem{lemma}[theorem]{Lemma}
\newtheorem{corollary}[theorem]{Corollary}

\newcounter{rrule}
\newenvironment{rrule}{\refstepcounter{rrule}\par\smallskip\noindent
\textbf{R\arabic{rrule}}\quad}{} %\par\noindent\rule{\textwidth}{1pt}}

\newcounter{brule}
\newenvironment{brule}{\refstepcounter{brule}\par\smallskip\noindent
\textbf{B\arabic{brule}}\quad}{} %\par\noindent\rule{\textwidth}{1pt}}

\newcounter{frrule}
\newenvironment{frrule}{\refstepcounter{frrule}\par\smallskip\noindent
\textbf{FR\arabic{frrule}}\quad}{} %\par\noindent\rule{\textwidth}{1pt}}

\newcounter{fbrule}
\newenvironment{fbrule}{\refstepcounter{fbrule}\par\smallskip\noindent
\textbf{FB\arabic{fbrule}}\quad}{} %\par\noindent\rule{\textwidth}{1pt}}

\newcommand{\lpvcalgname}{\mathrm{LPVCalg}}
\newcommand{\lpvcalg}[2]{\lpvcalgname(#1,#2)}

\newcommand{\F}[2]{\mathcal{F}_{#1,#2}}
\newcommand{\Fv}{\F{v}{k}}
\newcommand{\Falgname}{\mathrm{Falg}}
\newcommand{\Falg}[2]{\Falgname(#1,#2)}

\newcommand{\tree}[3]{T_{#1}(#2,#3)}
\newcommand{\treetop}[4]{T^{#1}_{#2}(#3,#4)}

\newcommand{\Ftree}[2]{\tree{\Falgname}{#1}{#2}}
\newcommand{\Ftreetop}[3]{\treetop{#1}{\Falgname}{#2}{#3}}

\newcommand{\lpvctree}[2]{\tree{\lpvcalgname}{#1}{#2}}
\newcommand{\combinedtree}[2]{\tree{}{#1}{#2}}
\newcommand{\combinedtreeb}[2]{\tree{\mathrm{C}}{#1}{#2}}

\newcommand{\timea}{$O^*(3.945^k)$}
\newcommand{\timeb}{$O^*(4.947^k)$}
\newcommand{\timec}{$O^*(5.951^k)$}

\begin{document}

\title{$l$-path vertex cover is easier than $l$-hitting set for small $l$}
\author{Dekel Tsur%
\thanks{Ben-Gurion University of the Negev.
Email: \texttt{dekelts@cs.bgu.ac.il}}}
\date{}
\maketitle

\begin{abstract}
In the \emph{$l$-path vertex cover problem}
the input is an undirected graph $G$ and an integer $k$.
The goal is to decide whether there is a set of vertices $S$ of
size at most $k$ such that $G-S$ does not contain a path with $l$ vertices.
In this paper we give parameterized algorithms for $l$-path vertex cover
for $l = 5,6,7$,
whose time complexities are \timea, \timeb, and \timec, respectively.
%Our algorithm is faster than previous algorithm for this problem.
\end{abstract}


\section{Introduction}
For an undirected graph $G$, an \emph{$l$-path} is a path in $G$ with $l$
vertices.
An \emph{$l$-path vertex cover} is a set of vertices $S$ such that
$G-S$ does not contain an $l$-path.
%every $l$-path in $G$ contains at least one vertex of $S$.
In the \emph{$l$-path vertex cover problem},
the input is an undirected graph $G$ and an integer $k$. The goal is to decide
whether there is an $l$-path vertex cover of $G$ with size at most $k$.
The problem for $l=2$ is the famous vertex cover problem.
For every fixed $l$, there is a simple reduction from vertex cover
to $l$-path vertex cover.
Therefore, $l$-path vertex cover is NP-hard for every constant $l \geq 2$.

For every fixed $l$, the $l$-path vertex cover problem is a special case
of the $l$-hitting set problem.
Therefore, there is a simple algorithm for $l$-path vertex cover with running
time $O^*(l^k)$.
Better time complexities can be achieved by using the best known algorithms
for $l$-hitting set.
For $l = 3,4,5,6$, the best algorithms for $l$-hitting set have running times of
$O^*(2.076^k)$, $O^*(3.076^k)$, $O^*(4.076^k)$, and $O^*(5.065^k)$,
respectively~\cite{wahlstrom2007algorithms,fomin2010iterative,fernau2006parameterized}.

For specific values of $l$,
it is possible to obtain algorithms for $l$-path vertex cover that are faster
than the best known algorithms for $l$-hitting set.
Algorithms for 3-path vertex cover were given
in~\cite{tu2015fixed,wu2015measure,katrenivc2016faster,chang2016fixed,xiao2017kernelization,Tsur_3pvc},
algorithms for 4-path vertex cover were given in~\cite{tu2016fpt,Tsur_4pvc},
and an algorithm for 5-path vertex cover was given in~\cite{cerveny2019}.

In this paper we give algorithms for $l$-path vertex cover for $l = 5,6,7$.
The time complexities of our algorithms are \timea, \timeb, and \timec,
respectively.
See Table~\ref{tab:results} for a comparison between our results and previous
results.
We note that our algorithm for $5$-path vertex cover is both faster than the
algorithm of {\v{C}}erven{\`y} and Ond{\v{r}}ej~\cite{cerveny2019}
and significantly simpler.
%Our algorithm is based on a general framework for solving vertex deletion problems, which is adapted from the algorithm of
%Boral et al.~\cite{boral2016fast} for cluster vertex deletion.
%This framework may be useful for obtaining improved algorithms for other
%vertex deletion problem.

\begin{table}
\centering
\begin{tabular}{@{}llllll@{}}
\toprule
$l$ & 3 & 4 & 5 & 6 & 7 \\
\midrule
Old & $O^*(1.713^k)$~\cite{Tsur_3pvc} & $O^*(2.619^k)$~\cite{Tsur_4pvc} &
$O^*(4^k)$~\cite{cerveny2019} & $O^*(5.065^k)$~\cite{fernau2006parameterized} &
$O^*(6.044^k)$~\cite{fernau2006parameterized} \\
New &  &  & \timea & \timeb & \timec \\
\bottomrule
\end{tabular}
\caption{Comparison of old and new results for the $l$-path vertex cover
problem.
The old results for $l=6$ and $l=7$ are due to algorithms for $l$-hitting set.}
\label{tab:results}
\end{table}

%\section{Preliminaries}
%For set of vertices $S$ in a graph $G$, $G[S]$ is the subgraph
%of $G$ induced by $S$ (namely, $G[S]=(S,E\cap (S\times S))$).
%We also define $G-S = G[V\setminus S]$.

\section{The algorithm}\label{sec:algorithm}

In this section we present an algorithm for solving $l$-path vertex cover.
While our algorithm is similar to the algorithm of
Boral et al.~\cite{boral2016fast} for the cluster vertex deletion problem,
there are several important differences. In particular, the method used for
computing a $v,k$-family (to be defined below) is completely different in our
algorithm.
For the rest of this section we consider the $l$-path vertex cover for some
fixed $l$.

We first give several definitions.
For set of vertices $S$ in a graph $G$, $G[S]$ is the subgraph
of $G$ induced by $S$ (namely, $G[S]=(S,E\cap (S\times S))$).
We also define $G-S = G[V\setminus S]$.
For a path $P$ in $G$, $V(P)$ is the set of the vertices of $P$.

For a vertex $v$ in $G$, $C_v$ is the connected component of $G$ that contains
$v$, and $G_v = G[C_v]$.
A \emph{$v$-path} is an $l$-path that contains $v$.
A \emph{$v$-hitting set} is a set of vertices $X \subseteq C_v \setminus \{v\}$
such that $X \cap V(P) \neq \emptyset$ for every $v$-path $P$.
Denote by $\beta(G,v)$ the minimum size of a $v$-hitting set in $G$.
Two $v$-paths $P,P'$ are called \emph{intersecting} if $V(P) \neq V(P')$
and $V(P) \cap V(P')$ contains at least one vertex other than $v$.

A \emph{$v,k$-family} of a graph $G$ is a family $\mathcal{F}$ of $v$-hitting
sets such that
(1) Every set in $\mathcal{F}$ has size at most $k$, and
(2) If $(G,k)$ is a yes instance of $l$-path vertex cover,
there is a $l$-path vertex cover $S$ of $G$ of size at most $k$ such that
either $v\in S$ or $A \subseteq S$ for some set $A \in \mathcal{F}$.

The following lemma will be used below to prove the correctness of
the algorithm.
\begin{lemma}\label{lem:X}
Let $X$ be a $v$-hitting set.
There is a minimum size $l$-path vertex cover $S$ such that
either $v \notin S$ or $|X \setminus S| \geq 2$.
\end{lemma}
\begin{proof}
Let $S$ be a minimum size $l$-path vertex cover of $G$.
Assume that $v\in S$ and $|X \setminus S| \leq 1$ otherwise we are done.
Let $S' = (S \setminus \{v\}) \cup X$.
Note that $S'$ is also an $l$-path vertex cover of $G$
(Suppose conversely that there is an $l$-path $P$ in $G-S'$.
Since $G-S$ does not contain an $l$-path, we have that $P$ contains $v$.
Namely, $P$ is a $v$-path.
By definition, $X \cap V(P) \neq \emptyset$ and therefore
$S' \cap V(P) \neq \emptyset$, a contradiction).
Additionally, $|S'| \leq |S|$.
Thus, $S'$ is a minimum size $l$-path vertex cover of $G$ and $v \notin S'$.
\end{proof}
\begin{corollary}\label{cor:beta-1}
If $\beta(G,v) = 1$ then there is a minimum size $l$-path vertex cover $S$
such that $v\notin S$.
\end{corollary}
%\begin{proof}
%\end{proof}

We now describe the main algorithm, which is called $\lpvcalgname$.
Algorithm $\lpvcalgname$ is a branching algorithm.
Given an instance $(G,k)$ the algorithm applies the first applicable rule from
the rules below.
The reduction rules of the algorithm are as follows.

\begin{rrule}
If $k < 0$, return `no'.
\end{rrule}

\begin{rrule}
If $G$ is an empty graph, return `yes'.
\end{rrule}

\begin{rrule}
If there is a vertex $v$ such that there are no $v$-paths,
return $\lpvcalg{G-v}{k}$.\label{rrule:redundant-vertex}
%return $(G-v,k)$.
%delete the vertex $v$.
\end{rrule}

\begin{rrule}
If there is a vertex $v$ such that $G_v-v$ does not contain an $l$-path,
return $\lpvcalg{G-C}{k-1}$.
\label{rrule:size-1}
%return $(G-C,k-1)$.
%delete the vertices of $C$ and decrease $k$ by 1.
\end{rrule}

If the reduction rules above cannot be applied, the algorithm chooses an
arbitrary vertex $v$.
Additionally, the algorithm decides whether $\beta(G,v)$ is 1, 2, or at least 3 
(this can be done in $n^{O(1)}$ time).
If $\beta(G,v) = 2$, the algorithm constructs a $v$-hitting set
$X = \{w_1, w_2\}$ of size~2.
Denote by $C_i$ the connected component of $G_v-v$ that contains $w_i$
(note that $C_1$ can be equal to $C_2$).

\begin{lemma}\label{lem:C0}
If reduction rule~R\ref{rrule:size-1} cannot be applied, $\beta(G,v) = 2$,
and the graphs $G[C_1]$ and $G[C_2]$ do not contain $l$-paths,
then there is exactly one connected component $C_0$ of $G_v-v$ that contains
an $l$-path.
\end{lemma}
\begin{proof}
Since reduction rule~R\ref{rrule:size-1} cannot be applied, there is at least
one connected component $C_0$ of $G_v-v$ that contains an $l$-path.
Suppose conversely that there is a connected component $C'_0 \neq C_0$ of
$G_v-v$ that contains an $l$-path.
Let $y_1,\ldots,y_l$ be an $l$-path in $G[C_0]$. Let $v,x_1,\ldots,x_s,y_i$ be
a shortest path in $G$ between $v$ and some vertex from $y_1,\ldots,y_l$.
We assume without loss of generality that $i \leq \lceil l/2 \rceil$.
We have that $x_1,\ldots,x_s,y_i,\ldots,y_l$ is a path in $G[C_0]$ with at
least $l/2$ vertices whose first vertex is adjacent to $v$.
Similarly, there is a path in $G[C'_0]$ with at
least $l/2$ vertices whose first vertex is adjacent to $v$.
Therefore, there is a $v$-path $P$ such that
$V(P) \subseteq \{v\} \cup C_0 \cup C'_0$.
Since $G[C_1]$ and $G[C_2]$ do not contain $l$-paths, it follows that
$V(P) \cap X = \emptyset$, contradicting the fact that $X$ is a $v$-hitting set.
\end{proof}

The algorithm now tries to apply the following reduction rules.

\begin{rrule}
If $\beta(G,v) = 2$, the graphs $G[C_1]$ and $G[C_2]$ do not contain $l$-paths,
and there is a $v$-path that does not contain a vertex from $C_0$,
return $\lpvcalg{G-v}{k-1}$.
\label{rrule:beta-2}
%return $(G-v,k-1)$.
\end{rrule}

Rule~R\ref{rrule:beta-2} is safe since there is a minimum size $l$-path vertex
cover that contains $v$: If $S$ is a minimum size $l$-path vertex cover that
does not contain $v$,
then $S$ must contain a vertex $x \in V(P)\setminus \{v\}$, where $P$ is
a $v$-path that does not contain a vertex from $C_0$.
By Lemma~\ref{lem:C0}, the connected component of $G_v-v$ that contains $x$ does not contain an $l$-path.
Therefore, the set $S' = (S\setminus \{x\}) \cup \{v\}$ is an $l$-path vertex
cover of $G$.
Since $|S'|=|S|$, we obtain that $S'$ is a minimum size $l$-path vertex cover
that contain $v$.

\begin{rrule}
If $\beta(G,v) = 2$ and the graphs $G[C_1]$ and $G[C_2]$ do not contain
$l$-paths, return $\lpvcalg{G'}{k}$ where the graph $G$ is obtained from $G$
as follows.
Let $r$ be the maximum length of a path that starts at $v$ in $G_v - C_0$.
Delete the vertices of $C_v \setminus (C_0 \cup \{v\})$ from $G$,
add $r$ new vertices to the graph, construct an $r$-path on the new vertices,
and add an edge between the $v$ and the first vertex in the $r$-path.
\label{rrule:beta-2b}
\end{rrule}

The safeness of Rule~R\ref{rrule:beta-2b} also follows from Lemma~\ref{lem:C0}.
When the reduction rules above cannot be applied,
the algorithm computes a $v,k$-family $\Fv$ (we will describe below how to
compute $\Fv$).
It then applies one of the following branching rules, depending on $\beta(G,v)$.

\begin{brule}
If $\beta(G,v) = 1$, branch on every set in $\Fv$.
\label{brule:1}
\end{brule}

The safeness of Rule~\ref{brule:1} follows from Corollary~\ref{cor:beta-1}.

\begin{brule}
If $\beta(G,v) = 2$,
let $i$ be an index such that $G[C_i]$ contains an $l$-path.
Construct a $(w_i,k-1)$-family $\F{w_i}{k-1}^{G-v}$ for the graph $G-v$.
Branch on every set in $\Fv$ and on $\{v\} \cup S$ for every
$S \in \F{w_i}{k-1}^{G-v}$.
\label{brule:2}
\end{brule}

Note that since Rule~R\ref{rrule:beta-2b} cannot be applied, the index $i$
exists.
By Lemma~\ref{lem:X} there is a minimum size $l$-path vertex cover $S$ such that
either $v \notin S$ or $w_i \notin S$.
Therefore Rule~B\ref{brule:2} is safe.

\begin{brule}
If $\beta(G,v) \geq 3$ branch on $\{v\}$ and on every set in $\Fv$.
\label{brule:3}
\end{brule}

We now describe an algorithm, called $\Falgname$, for constructing
a $v,k$-family for $G$.
We will show in Section~\ref{sec:no-intersecting} that if there are no two intersecting $v$-paths then $G$ has a simple structure and thus finding a
$v,k$-family for $G$ can be done in polynomial time.
Moreover, the $v,k$-family in this case consists of a single set.

Algorithm $\Falgname$ consists of the following reduction and branching rules.
\begin{frrule}
If $k < 0$, return $\emptyset$.
\end{frrule}
\begin{frrule}
If there are no $v$-paths in $G$, return $\{\emptyset\}$.
\end{frrule}
\begin{frrule}
If there are no intersecting $v$-paths,
compute a $v,k$-family for $G$ and return it.\label{frrule:compute}
\end{frrule}

\begin{fbrule}
Otherwise, let $P$ and $P'$ be intersecting $v$-paths.
Perform the following steps.
\vspace{-20pt}
\begin{algtab}
 $L \gets \emptyset$.\\
 \algforeach{$a \in (V(P)\cap V(P'))\setminus \{v\}$}
  \algforeach{$X \in \Falg{G-a}{k-1}$}
   Add $X\cup \{a\}$ to $L$.\\
  \algend
 \algend
 \algforeach{$a \in V(P) \setminus (V(P')\cup \{v\})$ and
			$a'\in V(P') \setminus (V(P)\cup \{v\})$ }
  \algforeach{$X \in \Falg{G-\{a,a'\}}{k-2}$}
   Add $X\cup \{a,a'\}$ to $L$.\\
  \algend
 \algend
 \algreturn $L$.\\
\end{algtab}
\end{fbrule}

\section{Analysis}
In this section we analyze the running time of algorithm $\lpvcalgname$.
The analysis is similar to the analysis of the algorithm of
Boral et al.~\cite{boral2016fast} (see also~\cite{Tsur_cluster}).

Let $A$ be some parameterize algorithm on graphs.
The run of algorithm $A$ on an input $(G,k)$ can be represented
by a \emph{recursion tree}, denoted $\tree{A}{G}{k}$, as follows.
The root $r$ of the tree corresponds to the call $A(G,k)$.
If the algorithm terminates in this call, the root $r$ is a leaf.
Otherwise, suppose that the algorithm is called recursively on the instances
$(G_1,k-a_1),\ldots,(G_t,k-a_t)$.
In this case, the root $r$ has $t$ children. The $i$-th child of $r$
is the root of the tree $\tree{A}{G_i}{k-a_i}$.
The edge between $r$ and its $i$-th child is labeled by $a_i$.
See Figure~\ref{fig:recursion-tree} for an example.

We define the \emph{weighted depth} of a node $x$ in $\tree{A}{G}{k}$ to be
the sum of the labels of the edges on the path from the root to $x$.
For an internal node $x$ in $\tree{A}{G}{k}$ define the
\emph{branching vector} of $x$ to be a vector containing
the labels of the edges between $x$ and its children.
The \emph{branching number} of a vector $(a_1,\ldots,a_t)$
(where $t \geq 2$) is the largest root of $P(x) = 1-\sum_{i=1}^t x^{-a_i}$.
The branching number of a node $x$ in $\tree{A}{G}{k}$ is the
branching number of the branching vector of $x$.
The running time of algorithm $A$ can be bounded by bounding the number of
leaves in $\tree{A}{G}{k}$.
The number of leaves in $\tree{A}{G}{k}$ is $O(c^k)$, where $c$ is the maximum
branching number of a node in the tree.

From the previous paragraph, we can bound the number of leaves in
$\tree{A}{G}{k}$ by giving an upper bound on the branching number of a node in
$\tree{A}{G}{k}$.
In some cases, there may be few nodes in $\tree{A}{G}{k}$ with large
branching numbers.
One can handle this case by modifying the tree $\tree{A}{G}{k}$ by contracting
some of its edges.
If $x$ is a node in $\tree{A}{G}{k}$ and $y$ is a child of $x$ which is an
internal node,
\emph{contracting} the edge $(x,y)$ means deleting the node $y$
and replacing every edge $(y,z)$ between $y$ and a child $z$ of $y$
with an edge $(x,z)$.
The label of $(x,z)$ is equal to the label of $(x,y)$ plus the label of $(y,z)$.
See Figure~\ref{fig:contract}.
Note that after edge contractions, the number of leaves in the contracted tree
is equal to the number of leaves in $\tree{A}{G}{k}$.
However, the maximum branching number of a node in the contracted tree may
be smaller than the maximum branching number of a node in $\tree{A}{G}{k}$.

For an integer $\alpha$, we define the \emph{top recursion tree} $\treetop{\alpha}{A}{G}{k}$ to be the tree
obtained by taking the subtree of $\tree{A}{G}{k}$ induced by the
nodes with weighted depth less than $\alpha$ and their children.
See Figure~\ref{fig:top} for an example.

\begin{figure}
\centering
\subfigure[$\tree{A}{G}{k}$\label{fig:recursion-tree}]{\includegraphics{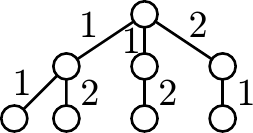}}
\quad
\subfigure[Contraction\label{fig:contract}]{\quad\includegraphics{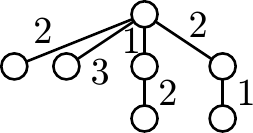}}
\quad
\subfigure[$\treetop{2}{A}{G}{k}$\label{fig:top}]{\includegraphics{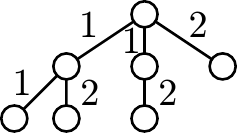}}
\caption{
Figure~(a) shows a recursion tree $\tree{A}{G}{k}$.
In this example, the call $A(G,k)$ makes three recursive calls:
$A(G_1,k-1)$, $A(G_2,k-1)$, and $A(G_3,k-2)$.
Figure~(b) shows the tree $\tree{A}{G}{k}$ after contracting the edge between
the root and its leftmost child.
The top recursion tree $\treetop{2}{A}{G}{k}$ is shown in
Figure~(c).}
\end{figure}

When we consider the recursion tree $\Ftree{G}{k}$ of algorithm $\Falgname$,
we assume that if in the recursive call on an instance $(G',k')$ the algorithm
applies Rule~FR\ref{frrule:compute} then the node $x$ in $\Ftree{G}{k}$ that
corresponds to this call has a single child $x'$ (which is a leaf) and the label
of the edge $(x,x')$ is the cardinality of the single set in the computed family.

To analyze the algorithm we define a tree $\combinedtree{G}{k}$ that
represents the recursive calls to both $\lpvcalgname$ and $\Falgname$.
Consider a node $x$ in $\lpvctree{G}{k}$, corresponding to a recursive call
$\lpvcalg{G'}{k'}$.
Suppose that in the recursive call $\lpvcalg{G'}{k'}$, the algorithm applies
Rule~B\ref{brule:3}.
Recall that in this case, the algorithm branches on $\{v\}$
and on every set in $\F{v}{k'}$. Denote $\F{v}{k'} = \{X_1,\ldots,X_t\}$.
In the tree $\lpvctree{G}{k}$, $x$ has $t+1$ children $y,x_1,\ldots,x_t$
(the child $y$ corresponds to the recursive call on
$(G'-v,k'-1)$, and a child $x_i$ corresponds to the recursive
call on $(G'-X_i,k'-|X_i|)$).
The label of the edge $(x,y)$ is $1$,
and the label of an edge $(x,x_i)$ is $|X_i|$.
The tree $\combinedtree{G}{k}$ also contains the nodes $x,y,x_1,\ldots,x_t$.
In $\combinedtree{G}{k}$, $x$ has two children $y$ and $x'$, where $x'$ is a new
node.
The labels of the edges $(x,y)$ and $(x,x')$ are $1$ and $0$,
respectively.
The node $x'$ is the root of a copy of the tree $\Ftree{G'}{k'}$
and the nodes $x_1,\ldots,x_t$ are the leaves of this tree.
Note that the label of an edge $(x,x_i)$ in the tree $\lpvctree{G}{k}$ is equal
to the sum of the labels of the edges on the path from $x$ to $x_i$ in
the tree $\combinedtree{G}{k}$.

%See Figure~\ref{fig:B3-tree} for an example.

Similarly, if in the recursive call $\lpvcalg{G'}{k'}$ that corresponds to $x$
the algorithm applies Rule~B\ref{brule:2}, then the algorithm branches on
every set in $\F{v}{k'}$ and on $\{v\} \cup S$ for every
$S \in \F{w_i}{k'-1}^{G-v}$.
Denote $\F{v}{k'} = \{X_1,\ldots,X_t\}$ and
$\F{w_i}{k'-1}^{G-v} = \{Y_1,\ldots,Y_s\}$.
In the tree $\lpvctree{G}{k}$, $x$ has $s+t$ children
$y_1,\ldots,y_s,x_1,\ldots,x_t$.
The label of an edge $(x,x_i)$ is $|X_i|$, and the label of an edge $(x,y_i)$
is $1+|Y_i|$.
In the tree $\combinedtree{G}{k}$, $x$ has two children $y'$ and $x'$.
The labels of the edges $(x,y')$ and $(x,x')$ are $1$ and $0$, respectively.
The node $x'$ is the root of a copy of the tree $\Ftree{G'}{k'}$, and
$x_1,\ldots,x_t$ are the leaves of this tree.
The node $y'$ is the root of a copy of the tree $\Ftree{G'-v}{k'-1}$,
and $y_1,\ldots,y_s$ are the leaves of this tree.

Finally,
suppose that in the recursive call $\lpvcalg{G'}{k'}$ that corresponds to $x$,
the algorithm applies Rule~B\ref{brule:1}.
In this case, the algorithm branches on every set in $\F{v}{k'}$.
In the tree $\lpvctree{G}{k}$, $x$ has $t = |\F{v}{k'}|$ children
$x_1,\ldots,x_t$.
In $\combinedtree{G}{k}$, the node $x$ is the root of a copy of the tree
$\Ftree{H'}{k'}$, and $x_1,\ldots,x_t$ are the leaves of this tree.

Our goal is to analyze the number of leaves in $\combinedtree{G}{k}$.
For this purpose, we perform edge contractions on $\combinedtree{G}{k}$ to
obtain a tree $\combinedtreeb{G}{k}$.
Consider a node $x$ in $\combinedtree{G}{k}$ that corresponds to
a recursive call $\lpvcalg{G'}{k'}$.
Suppose that in the recursive call $\lpvcalg{G'}{k'}$, the algorithm applies
Rule~B\ref{brule:3}.
Using the same notations as in the paragraphs above, we contract the following
edges:
\begin{enumerate}
\item
The edge $(x,x')$.
\item
The edges of the copy of $\Ftree{G'}{k'}$ whose endpoints correspond to
internal nodes in the tree $\Ftreetop{3}{G'}{k'}$.
In other words, the two endpoints of such edges have weighted depths less than
$3$ in the copy of $\Ftree{G'}{k'}$.
\end{enumerate}
If in the recursive call $\lpvcalg{G'}{k'}$ the algorithm applies
Rule~B\ref{brule:2}, we contract the following edges.
\begin{enumerate}
\item
The edges $(x,y')$ and $(x,x')$.

\item
The edges of the copy of $\Ftree{G'}{k'}$ whose endpoints correspond to
internal nodes in the tree $\Ftreetop{2}{G'}{k'}$.
\end{enumerate}
If in the recursive call $\lpvcalg{G'}{k'}$ the algorithm applies
Rule~B\ref{brule:1}, we do not contract edges.
%Therefore, the branching vector of $x$ in this case is at least $(1)$ or
%at least $(1,2)$.

The nodes in $\combinedtree{G}{k}$ and $\combinedtreeb{G}{k}$ that correspond
to nodes in $\lpvctree{G}{k}$ are called \emph{primary nodes},
and the remaining nodes are \emph{secondary nodes}.

If $x$ is a secondary node of $\combinedtree{G}{k}$ with more than one child,
the branching vector of $x$ is of the form $V_s = (1,\ldots,1,2,\ldots,2)$,
where the value 1 appears $s$ times for some $s \leq l-2$ and
the value 2 appears $(l-1-s)^2$ times.
The vector among $V_1,\ldots,V_{l-2}$ with largest branching number is $V_1$.
For $l=4$, for example, $V_1 = (1,2,2,2,2)$ and the branching number of $V_1$
is less than 2.562.
Therefore the branching number of all the secondary nodes in
$\combinedtreeb{G}{k}$ and all the primary nodes that correspond to recursive
calls in which Rule~B\ref{brule:1} is applied have branching numbers
less than 2.562.

We now bound the branching numbers of the remaining primary nodes of
$\combinedtreeb{G}{k}$.
Let $x$ be a primary node and suppose that in the corresponding recursive call
$\lpvcalg{G'}{k'}$ the algorithm applies Rule~B\ref{brule:3}.
Using the same notations as in the paragraphs above,
the branching vector of $x$ is $(1,c_1,\ldots,c_t)$,
where $(c_1,\ldots,c_t)$ are the weighted depths of the leaves of the top
recursion tree $\Ftreetop{3}{G'}{k'}$.
The top recursion tree that gives the worst branching number is the tree in
which every internal node has branching vector $V_1$.
See Figure~\ref{fig:B3-tree}.
For $l = 4$, the branching vector of $x$ in this case is
\[ (1, 3, 4, 4, 4, 4, 3, 3, 3, 3, 3, 4, 4, 4, 4, 3, 4, 4, 4, 4, 3, 4, 4, 4, 4, 3, 4, 4, 4, 4) \]
and the branching number is less than 2.897.

\begin{figure}
\centering
\includegraphics{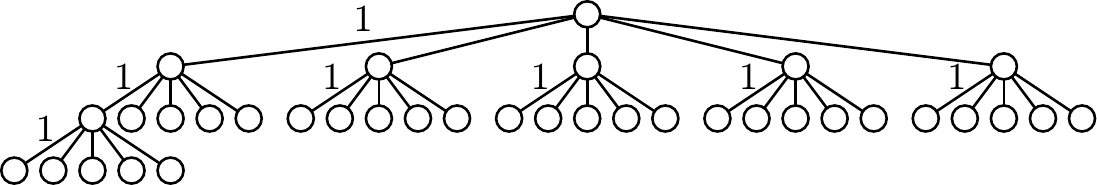}
\caption{The worst top recursion tree for $l = 4$.
The branching vector of every internal node is $(1,2,2,2,2)$.
Edges whose labels are not shown have label~2.\label{fig:B3-tree}}
\end{figure}

If in the recursive call $\lpvcalg{G'}{k'}$ the algorithm applies
Rule~B\ref{brule:2} then the branching vector of $x$ has the form
$(1+a,c_1,\ldots,c_t)$ for $a\geq 1$, or the form
$(2,\ldots,2,3,\ldots,3,c_1,\ldots,c_t)$, where
where $(c_1,\ldots,c_t)$ are the weighted depths of the leaves of the top
recursion tree $\Ftreetop{2}{G'}{k'}$,
and in the second form the value 2 appears $s$ times for some $s \leq l-2$ and
the value 3 appears $(l-1-s)^2$ times.
The worst case is when the branching vector is of the form
$(2,3,\ldots,3,c_1,\ldots,c_t)$, where the value 3 appears $(l-2)^2$ times.
Additionally, the top recursion tree that gives the worst branching number is
the tree in which every internal node has branching vector $V_1$.
For $l = 4$, the branching vector of $x$ in this case is
\[ (2, 3, 3, 3, 3, 2, 3, 3, 3, 3, 2, 2, 2, 2) \]
and the branching number is less than 2.952.

It follows that for $l = 4$, all nodes of $\combinedtreeb{G}{k}$ have
branching numbers less than 2.952.
Therefore, the time complexity of the algorithm for $l = 4$ is $O^*(2.952^k)$.
Similarly, for $l = 5,6,7$, the time complexities of the algorithm are
\timea, \timeb, and \timec, respectively.

\section{Graphs with no intersecting $v$-paths}
\label{sec:no-intersecting}

In this section we show how to compute a $v,k$-family in a graph $G$ that
does not contain two intersecting $v$-paths.
We will show that when $l \leq 7$ the graph $G$ has a simple structure,
and we will use this fact to obtain an algorithm for computing a $v,k$-family.

A $v$-path $P = x'_1,\ldots,x'_p,v,\allowbreak x_1,\ldots,x_{l-1-p}$
(where $p$ can be 0)
is called \emph{canonical} if $p \leq \lfloor (l-1)/2\rfloor$ and 
there is no path $P' = y'_1,\ldots,y'_q,v,y_1,\ldots,y_{l-1-q}$
such that $V(P') = V(P)$ and $q < p$.

\begin{lemma}\label{lem:same-vertex-set}
Let $G$ be a graph that does not contain two intersecting $v$-paths.
Then, all $v$-paths in $G$ have the same vertex set.
\end{lemma}
\begin{proof}
Suppose conversely that $P$ and $P'$ are canonical $v$-paths such that
$V(P) \neq V(P')$.
Since $P$ and $P'$ do not intersect, $V(P) \cap V(P') = \{v\}$.
Denote $P = x'_1,\ldots,x'_p,v,\allowbreak x_1,\ldots,x_{l-1-p}$ and
$P' = y'_1,\ldots,y'_q,v,y_1,\ldots,y_{l-1-q}$.
The path $P'' = x_{l-1-p},\ldots,x_1,v,y_1,\ldots,y_{l-1-q}$ is a path
with at least $l$ vertices.
Thus, the first $l$ vertices of $P''$ is a $v$-path that intersects $P$,
a contradiction.
\end{proof}

For the rest of this section, assume that $G$ is a graph that does not contain
two intersecting $v$-paths.
We also assume that $P = x'_1,\ldots,x'_p,v,x_1,\ldots,x_{l-1-p}$ is a
cannonical $v$-path in $G$.

\subsection{$l \leq 6$}
We now consider the case $l \leq 6$.
We will discuss the case $l = 7$ in the next subsection.

\begin{lemma}\label{lem:one-neighbor}
If $l \leq 6$, for every connected component $C$ of $G_v-V(P)$ there is exactly
one vertex in $V(P)$ that has neighbors in $C$.
\end{lemma}
\begin{proof}
Denote $x_0 = v$.
Note that $l\geq 6$ implies that $p \leq 2$.
Therefore, a vertex $x'_i$ does not have neighbors in $V(G) \setminus V(P)$
(otherwise, if $x'_i$ is adjacent to $y\in V(G) \setminus V(P)$,
the path $y,x'_i,\ldots,x'_p,v,x_1,\ldots,x_{l-1-p}$ is a path with at least
$l$ vertices, and therefore there is a $v$-path that intersects $P$,
a contradiction).

Suppose conversely that there is a connected component $C$ of $G_v-V(P)$
such that there are at least two vertices in $V(P)$ that have neighbors in $C$.
By the above, these vertices are from $\{x_0,\ldots,x_{l-1-p}\}$.
Suppose that $x_i$ and $x_j$ have neighbors in $C$, where $i < j$.
Since $C$ is a connected component of $G_v-V(P)$,
there is a path $x_i,y_1,\ldots,y_s,x_j$ in $G$ such that
$y_1,\ldots,y_s \in C$.
If $s > j-i-2$ then the path
$x'_1,\ldots,x'_p,v,x_1,\ldots,x_i,y_1,\ldots,y_s,x_j,\ldots,x_{l-1-p}$
has $l-(j-i-1)+s \geq l$ vertices,
and therefore there is a $v$-path that intersects $P$, a contradiction.
Therefore, $s \leq j-i-2$.
If either $i > 0$ or $p > 0$, the path
$x'_1,\ldots,x'_p,v,x_1,\ldots,x_i,y_1,\ldots,y_s,x_j,\ldots,x_{i+1}$
has $p+1+j+s \geq p+1+(s+i+2)+s \geq p+i+5 \geq 6 \geq l$ vertices,
and therefore there is a $v$-path that intersects $P$, a contradiction.
The remaining case is when $i = 0$ and $p = 0$.
In this case, $y_1,v,x_1,\ldots,x_{l-2}$ is a $v$-path that intersects $P$,
a contradiction.
\end{proof}

%\begin{lemma}\label{lem:x-l-1-p}
%Let $G$ be a graph that does not contain two intersecting $v$-paths and
%suppose that $l \leq 6$.
%Let $P = x'_1,\ldots,x'_p,v,x_1,\ldots,x_{l-1-p}$ be a canonical $v$-path.
%If $x_{l-1-p}$ has a neighbor not in $V(P)$ then $p = 0$ and the distance
%between $v$ and $x_{l-1}$ is $l-1$.
%\end{lemma}
%\begin{proof}
%Let $y \in V(G)\setminus V(P)$ be a neighbor of $x_{l-1-p}$.
%If $p > 0$ $x'_2,\ldots,x'_p,v,x_1,\ldots,x_{l-1-p},y$ is a $v$-path that
%intersects $P$, a contradiction.
%Thus, $p = 0$.
%Now suppose conversely that the distance between $v$ and $x_{l-1}$ is less than
%$l-1$.
%Let $i$ be the minimum index such that the distance between $v$ and $x_i$ is
%less than $i$.
%By Lemma~\ref{lem:one-neighbor}, the shortest path between $v$ and $x_i$ must be
%of the form $v,x_1,\ldots,x_{i-2},x_i$.
%Therefore, $v,x_1,\ldots,x_{i-2},x_i,\ldots,x_{l-1},y$ is a $v$-path that
%intersects $P$, a contradiction.
%\end{proof}

\begin{lemma}\label{lem:no-Pl}
If $l \leq 6$ and $C$ is a connected component of $G_v-V(P)$ that contains
an $l$-path then the unique vertex in $V(P)$ that has neighbors in $C$
is $x_{l-1-p}$.
\end{lemma}
\begin{proof}
Let $x$ be the unique vertex in $V(P)$ that has neighbors in $C$.
Let $P' = y_1,\ldots,y_l$ be an $l$-path in $C$.
Let $y_q$ be the vertex in $V(P')$ that is closest to $x$, and suppose
without loss of generality that $q \leq \lceil l/2 \rceil$.
Let $x,z_1,\ldots,z_s,y_q$ be the shortest path from $x$ to $y_q$
($s$ can be 0).

We now show that if $i < l-1-p$ then $x \neq x_i$.
Suppose conversely that $x = x_i$.
The path $v,x_1,\ldots,x_i,z_1,\ldots,z_s,y_q,\ldots,y_l$
has at least $2+\lceil (l+1)/2\rceil \geq l$ vertices (since $l\leq 6$),
and therefore there is a $v$-path that intersects $P$, a contradiction.
Therefore, $x \neq x_i$ for every $i < l-1-p$.
Similarly, $x \notin \{x'_1,\ldots,x'_p,v\}$.
\end{proof}

For a $v$-hitting set $X$, let $\alpha(x)$ be the maximum integer $q$ such that
there is a $q$-path in $G-X$ whose first vertex is $x_{l-1-p}$ and its second
vertex is $x_{l-2-p}$.
If $x_{l-1-p} \in X$ then define $\alpha(X) = 0$ and if
$x_{l-1-p} \notin X$ and $x_{l-2-p} \in X$ then define $\alpha(X) = 1$.
A minimum size $v$-hitting set $X$ is called \emph{optimal} if
$\alpha(X) \leq \alpha(Y)$ for every minimum size $v$-hitting set $Y$.

\begin{lemma}\label{lem:optimal}
If $l \leq 6$, $X$ is an optimal $v$-hitting set, and $|X|\leq k$
then $\mathcal{F} = \{X\}$ is a $v,k$-family.
\end{lemma}
\begin{proof}
Suppose that $S$ is an $l$-path vertex cover of $G$ of size at most $k$.
If $v\in S$ we are done. Suppose that $v \notin S$.
Let $S_0$ be the set of all vertices in $S$ that are contained in $v$-paths
of $G$.
Since $X$ is a minimum size $v$-hitting set, $|S_0| \geq |X|$.
By Lemma~\ref{lem:no-Pl} and the optimality of $X$, if $|S_0| = |X|$ then
$(S\setminus S_0) \cup X$ is an $l$-path vertex cover of $G$ of size at most
$k$.
Otherwise we have $|S_0| > |X|$ and therefore
$(S\setminus S_0) \cup X \cup \{x_{l-1-p}\}$ is an $l$-path vertex cover of $G$
of size at most $k$.
\end{proof}

\begin{lemma}\label{lem:v-hitting-set}
If $l \leq 6$, there is an optimal $v$-hitting set $X$ such that
$X \subseteq V(P)$.
\end{lemma}
\begin{proof}
Let $X$ be an optimal $v$-hitting set $X$.
If there is a vertex $x \in X \setminus V(P)$, let $C$ be the connected
component of $G_v-V(P)$ containing $x$ (note that $x\in C_v$ due to the
minimality of $X$).
Let $y$ be the unique vertex in $V(P)$ that has neighbors in $C$.
The set $X' = (X \setminus \{x\}) \cup \{y\}$ is also a $v$-hitting set.
Additionally, $|X'| = |X|$ and $\alpha(X') \leq \alpha(X)$.
By repeating this process we obtain an optimal $v$-hitting set that is
contained in $V(P)$.
\end{proof}

We now describe how to construct a $v,k$-family of a graph $G$ that does not
contain intersecting $v$-sets (assuming that $l\leq 6$).
%If $x_{l-1-p}$ does not have neighbors in $V(G)\setminus V(P)$,
%every $v$-hitting set is also a deletion set (by Lemma~\ref{lem:no-Pl}).
%By Lemma~\ref{lem:v-hitting-set}, we can find a minimum size $v$-hitting set $X$
%by going over all subsets of $V(P)$ and checking which subsets are $v$-hitting
%sets.
%The family $\mathcal{F} = \{ X \}$ is a $v,k$-family of $G$.
%Now suppose that $x_{l-1-p}$ has at least one neighbor in $V(G)\setminus V(P)$.
%By Lemma~\ref{lem:v-hitting-set}, there is a minimum size $v$-hitting set $X$
By enumerating all subsets of $V(P)$,
we can find an optimal $v$-hitting set $X \subseteq V(P)$.
By Lemma~\ref{lem:optimal}, the family $\mathcal{F} = \{ X \}$
is a $v,k$-family of $G$.

\subsection{$l = 7$}

We now describe the differences between the case $l = 7$ and the case
$l \leq 6$.

First, in contrast to Lemma~\ref{lem:no-Pl}, it is possible that there is
connected component $C$ of $G-V(P)$ such that $C$ contains an $l$-path
and there is a vertex $x \in V(P)\setminus \{x_{l-1-p}\}$ that is
adjacent to a vertex in $C$.
It can be shown that in this case, $x = x_1$.
Additionally, $x_1$ is adjacent to a single vertex $y$ in $C$, and every
$l$-path in $C$ must contain $y$.
It follows that if $S$ is an $l$-path vertex cover of $G$ and $y \notin S$ then
$(S\setminus C)\cup \{y\}$ is also an $l$-path vertex cover of $G$.

In contrast to Lemma~\ref{lem:one-neighbor},
it is possible that there is a connected component $C$ of $G_v-V(P)$ such that
there are two vertices in $V(P)$ that have neighbors in $C$.
It is easy to show that in this case the vertices in $V(P)$ that have neighbors
in $C$ are $x_1$ and $x_4$.
Additionally, $|C| = 1$ (otherwise there is a $v$-path $v,x_1,\ldots,x_4,y,y'$,
where $y,y' \in C$, and this $v$-path intersects $P$, a contradiction).
We call such a connected component an \emph{$x_1/x_4$-component}.
Note that there can be several $x_1/x_4$-components.

We now claim that Lemma~\ref{lem:v-hitting-set} is also true when $l = 7$.
To show that we need to consider the case when there are $x_1/x_4$-components.
Suppose that $C = \{y\}$ is an $x_1/x_4$-component.
We claim that every 7-path that contains $y$ also contains $x_4$.
Suppose conversely that $P'$ is a 7-path that contains $y$ and does not contain
$x_4$.
Then $P'$ must be of the form $y,x_1,\ldots,x_i,y_1,\ldots,y_{6-i}$ where
$i \leq 3$ and $y_1,\ldots,y_{6-i}$ are vertices in some connected component of
$G_v-V(P)$. The path $v,x_1,\ldots,x_i,y_1,\ldots,y_{6-i}$ is a $v$-path that
intersects $v$, a contradiction.
Therefore, the claim is true.
From the claim we have that if $X$ is an optimal $v$-hitting set then
$X' = (S \setminus \{y\}) \cup \{x_4\}$ is also an optimal $v$-hitting set.
Therefore, Lemma~\ref{lem:v-hitting-set} is also true when $l = 7$.

From the above discussion, we can construct a $v,k$-family of $G$ using the same
algorithm described above for the case $l\leq 6$.

\section{Concluding remarks}
We have shown algorithms for $l$-path vertex cover for $l = 5,6,7$ that are
faster than previous algorithms.
It may be possible to use our approach for $l = 8$ or other small values of $l$
by giving a more involved case analysis of graphs that do not have intersecting
$v$-paths.
An interesting open question is whether $l$-path vertex cover can be solved
in $O^*((l-1-\epsilon)^k)$ time for every value of $l$.

\bibliographystyle{plain}
\bibliography{lpvc,parameterized,dekel}

\end{document}